%% file: main.tex
\documentclass{article}
\usepackage{neurodata}

\title{Preserving Derivative Information while Transforming Neuronal Curves}

\author[1,2]{Thomas L. Athey\thanks{tathey1@jhu.edu}}
\author[3, 4]{Daniel J. Tward}
\author[5]{Ulrich Mueller}
\author[2,6,7,8]{Laurent Younes}
\author[1,2,7,8]{Joshua T. Vogelstein}
\author[1,2,7,8]{Michael I. Miller}

\affil[1]{Department of Biomedical Engineering, Johns Hopkins University, Baltimore, MD, USA}
\affil[2]{Institute of Computational Medicine, Johns Hopkins University, Baltimore, MD, USA}
\affil[3]{Department of Computational Medicine, University of California at Los Angeles, Los Angeles, CA, USA}
\affil[4]{Department of Neurology, University of California at Los Angeles, Los Angeles, CA, USA}
\affil[5]{Department of Neuroscience, Johns Hopkins University, Baltimore, MD, USA}
\affil[6]{Department of Applied Mathematics \& Statistics, Johns Hopkins University, Baltimore, MD, USA}
\affil[7]{Center for Imaging Science, Johns Hopkins University, Baltimore, MD, USA}
\affil[8]{Kavli Neuroscience Discovery Institute, Johns Hopkins University, Baltimore, MD, USA}

\begin{document}
\maketitle

\begin{abstract}
The international neuroscience community is building the first comprehensive atlases of brain cell types to understand how the brain functions from a higher resolution, and more integrated perspective than ever before. In order to build these atlases, subsets of neurons (e.g. serotonergic neurons, prefrontal cortical neurons etc.) are traced in individual brain samples by placing points along dendrites and axons. Then, the traces are mapped to common coordinate systems by transforming the positions of their points, which neglects how the transformation bends the line segments in between.  In this work, we apply the theory of jets to describe how to preserve derivatives of neuron traces up to any order. We provide a framework to compute possible error introduced by standard mapping methods, which involves the Jacobian of the mapping transformation. We show how our first order method improves mapping accuracy in both simulated and real neuron traces under random diffeomorphisms. Our method is freely available in our open-source Python package brainlit.
\end{abstract}

\section*{Main}

\input{0-intro}

\section*{Results}

\subsection*{Action of Diffeomorphisms on Discrete Samplings}

\input{1-action}
\label{sec:action}

\subsection*{Error Analysis of Zeroth and First Order Mapping}

\input{2-erroranalysis}
\label{sec:firstorder}

\subsection*{Software Implementation}

\input{3-softwaremethod}

\label{sec:software}

\subsection*{Application to Real Neurons}

\input{5-realneurons}

\label{sec:realneurons}

\section*{Discussion}

\input{6-discussion}

\section*{Methods}

\input{methods}

\clearpage

\section*{Data Availability}

The neuron traces used in this work were from the MouseLight project's NeuronBrowser website: https://ml-neuronbrowser.janelia.org/. Specifically, we used traces: AA-1087, 1089-1093, 1149, 1192, 1196, 1223, 1413, 1417, 1477, 1493, 1537-1540, 1543 and 1548.

\section*{Code Availability}

The code used in this work is available in our open-source Python package brainlit: http://brainlit.neurodata.io/.

\section*{Acknowledgements}

This work is supported by the National Institutes of Health grants RF1MH121539, P41EB015909, R01NS086888, U19AG033655, RO1AG066184-01 the National Science Foundation grants 2031985, 2014862, and the CAREER award and the National Institute of General Medical Sciences Grant T32GM119998. We thank the MouseLight team at HHMI Janelia for providing us with access to this data, and answering our questions about it.

\section*{Ethics declaration}

Under a license agreement between AnatomyWorks and the Johns Hopkins University, Dr. Miller and the University are entitled to royalty distributions related to technology described in the study discussed in this paper. Dr. Miller is a founder of and holds equity in AnatomyWorks. This arrangement has been reviewed and approved by the Johns Hopkins University in accordance with its conflict of interest policies. The remaining authors have no conflicts of interest to declare.

\section*{Author Contribution Statement}

T.L.A. developed the mathematical details and the software, executed the experiments and contributed to the writing of this paper. D.J.T. contributed to the experimental design, and writing of the paper. U.M. contributed to the conceptual development, and data availability. L.Y. formulated the mathematical framework, contributed to the experimental design and the writing of the paper. J.T.V. contributed to the scope and design of the experiments. M.I.M. contributed to the scope of the experiments and the writing of the paper.

\section*{Inclusion and Ethics Statement}

The purpose of this work is to investigate a computational technique that is being used in the brain mapping community. The roles and responsibilities of the authors were discussed prior to the research.  Efforts were made to make the work accessible by, for example, providing data used in the experiments, and maintaining an open-source repository of the code with extensive documentation.

\section*{Research Animals Statement}

The mouse projection neuron traces came from the MouseLight project's NeuronBrowser website and the experimental protocols that generated this data can be found in Winnubst et al. 2019 \cite{winnubst2019reconstruction}. 

\bibliographystyle{plainnat}
\bibliography{refs}

\section*{Supplement}
\setcounter{equation}{0}
\setcounter{proposition}{0}
\setcounter{statement}{0}
\input{supplement}

\end{document}

%% file: 0-intro.tex

The brain functions as a network of chemical and electrical activity, so identifying how neurons connect across brain regions is central to understanding how the brain works, and how to treat brain diseases. Modern neuroscience techniques can image single neuron morphology at scale \cite{economo2016platform}, and subsequent neuron tracing can help discover new morphological subtypes \cite{winnubst2019reconstruction}. Due to anatomical variation, and deformations that may have occurred during tissue preparation, neuron traces need to be mapped between coordinate spaces to compare morphologies from different brain samples. Brain registration software often includes neuron mapping implementations, but these implementations have not been thoroughly characterized from a numerical analysis perspective. 

This question is relevant to the ongoing work of the international neuroscience community, including the Brain Initiative Cell Census Network (BICCN), to establish comprehensive neuronal atlases of the mammalian brain \cite{brain2021multimodal}. This effort has produced many images of stained or fluorescently labeled brains, which are being used to generate digital neuron traces for morphological analysis. The traces are commonly stored as a set of connected 3D coordinates, or knots, such as in the SWC format \cite{stockley1993system, cannon1998line}. The connections between the knots are classically represented as cylinders \cite{cannon1998line}, or conical frustums \cite{o2020module}, but here we ignore radius information, since it is not generated by all neuron tracing methods. Consequently, the whole neuron trace is considered to be a tree of piecewise linear curves.

In order to assemble these traces into a complete picture of the various neuron morphologies in the brain, scientists need a way to map neuron traces into common coordinate systems. Several popular software applications exist for this task and are used to assemble atlases of neuron morphology. For example, Peng et al. \cite{peng2021morphological} used mBrainAligner \cite{qu2022cross}, Gao et al. \cite{gao2022single} used the Computational Morphometry Toolkit, and the MouseLight project \cite{winnubst2019reconstruction} used displacement fields from Fedorov et al. \cite{fedorov20123d}. Existing methods use what we call \textit{zeroth order} curve mapping in that they only map the positions of the knots. However, depending on the nonlinearity of the mapping, and the continuous representation of the neuron trace, zeroth order mapping is sensitive to different samplings of the original neuronal curve (Fig. \ref{fig:summary}a,b). In other words, sampling the same curve different ways while tracing in the original image may lead to different mapped morphologies. It is critical that neuron mapping methods preserve the geometry of digital neuron traces in order to build reliable atlases of neuron morphology, and to accurately identify deviations in diseased brains.

In this work, we introduce a method to preserve derivative information when mapping neuronal curves, and investigate the conditions under which this technique is advantageous to existing methods (Figure \ref{fig:summary}). We applied our method to both simulated data and real neuron traces from a whole mouse brain image, and the code used developed in this work is freely available in our Python package brainlit.

\begin{figure}[ht]
\centering
\includegraphics[width=\textwidth]{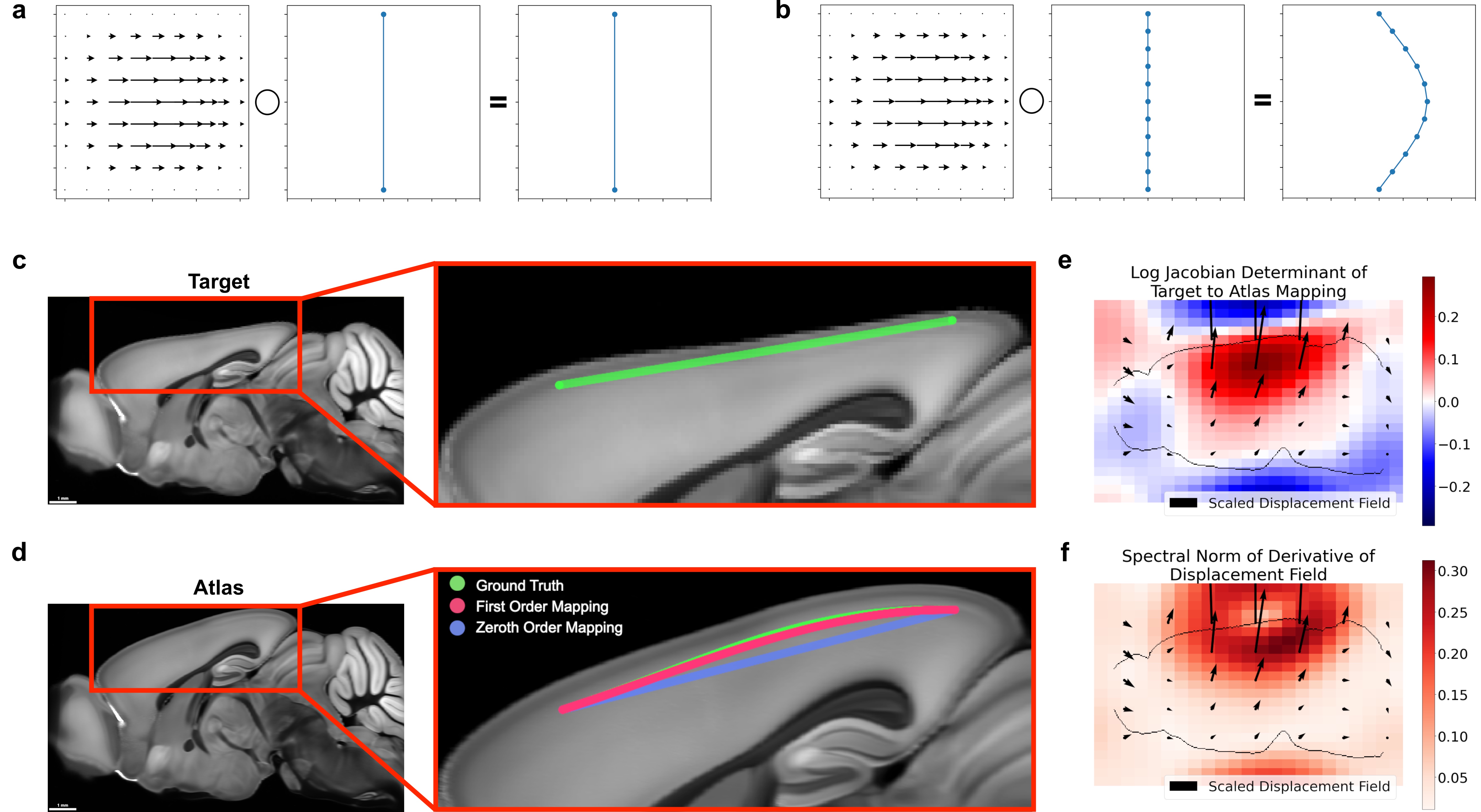}
\caption{Neglecting the action of a nonlinear mapping on a curve's derivatives can introduce errors. \textbf{a-b} Different samplings of a curve can lead to different results under nonlinear deformations, such as only sampling the endpoints (\textbf{a}) versus sampling several times along the curve (\textbf{b}). \textbf{c-d} Large distances between control points can contribute to mapping inaccuracies. The green line segment following cortical layers 2/3 in a synthetic mouse brain image (\textbf{c}) is defined only by its endpoints. Transforming only the positions of the endpoints (zeroth order mapping, \textbf{d}), is less accurate than incorporating the action on the derivatives as well (first order mapping, \textbf{d}). \textbf{e-f} Quantitative descriptions of the mapping from target to atlas via the logarithm of the Jacobian determinant, which quantifies expansion and compression (\textbf{e}), and the spectral norm of the displacement field, which plays a role in an error bound of zeroth order mapping (\textbf{f}).}
\label{fig:summary}
\end{figure}

%% file: 1-action.tex
In the following sections, we use $C^k$ to represent the space of continuous functions with $k$ continuous derivatives, where the domain and range can be inferred by the context. We model a neuronal branch (dendrite or axon) as a regular 3D curve $c:[0,L] \rightarrow \mathbb{R}^3$, $c \in C^k$, and $\vert \dot c \vert > 0$. When a neuronal curve is traced, it is typically stored as a sequence of points  $\{x_i=c(t_i): t_i < t_{i+1}\}_{i=1}^n$, where the independent variables $t_i$ can be taken to be the indices of the points. When there is a diffeomorphism between coordinate systems $\phi: \mathbb{R}^3\rightarrow \mathbb{R}^3$, these traces are mapped via the group action:

\begin{align*}
    \phi \cdot \{x_i\}_{i=1}^n &= \{\phi(x_i)\}_{i=1}^n
\end{align*}

We want to extend the space of traces, and the associated action, to include derivatives of the underlying curve denoted $\partial_t c$. This can be done using the jet space $J^k$. In our setting, $J^k=[0,L] \times X^{(k)}$, where an element of $X^{(k)}$ is a $k+1$-tuple $(x^0,x^1,...,x^k) \in (\mathbb R^3)^{k+1}$ representing a position and first $k$ derivatives of a curve in $\mathbb R^3$. A $C^k$ curve $c: [0,L] \to \mathbb R^3$ can be extended to a curve $\hat c: [0,L] \to X^{(k)}$ simply by adding derivatives, with $\hat c(t) = (c(t), \partial_t c(t), \ldots, \partial_t^k c(t)) \in X^{(k)}$ \cite{olver1995equivalence}. 

The $C^k$ diffeomorphisms have a natural group action on the jet space $J^k$, ensuring the commutation between the standard action of diffeomorphisms on curves, $(\phi, c) \mapsto \phi\circ c$ and their extensions, such that the identity $\phi\cdot \hat c(t) = \widehat{\phi\circ c}(t)$ holds for all curves $c$ and times $t$, defining the left-hand side. For example, for $k=2$, this provides
\[
\phi\cdot (t, x^{0}, x^{1}, x^{2}) = (t, \phi(x^{0}), D\phi(x^{0}) x^{1}, D\phi(x^{0}) x^{2} + D^2\phi(x^{0}) (x^{1}, x^{1}))
\]


Neuron traces, as mentioned before, involve a sequence of samples with time-stamps $\{(t_i,x^{(k)}_i)\}_{i=1}^n$, identified as elements of $(J^k)^n$, the $n$-fold Cartesian product of $J^k$. Our diffeomorphisms will act on such a sequence as follows:

\begin{statement}
For a sequence of time-stamped elements on the jet space, 
$T = \{(t_i, x_i^{(k)})\}_{i=1}^n$ in $(J^k)^n$, we define the action of diffeomorphisms


\begin{align}
    \phi \cdot T = \{(t_i, \phi\cdot x_i^{(k)})\}_{i=1}^n \label{eq:action}
\end{align}


\end{statement}

The fact that this operation provides an action is is an established result \cite{olver1995equivalence}, and the proof is provided in the Supplement. We will define \textit{$k$'th order discrete mapping} to be the action in Equation \ref{eq:action} of a diffeomorphism on a curve sampling that includes $k$ derivatives. The axioms that define group actions are important to verify because they ensure that applying the identity transformation does not change the object, and that applying a composition of transformations is equivalent to applying the individual transformations successively. Further, group actions can exchange mathematical structure between the acting group and the set being acted upon, and they are at the core of several important theorems \cite{suksumran2016gyrogroup}.

The $k$'th order discrete mapping method allows us to compute the first $k$ derivatives of the transformed curve. We will interpolate the transformed curve using splines of order $2k+1$ that satisfy the derivative values. For example, zeroth order mapping will produce a first order spline and first order mapping will produce a cubic Hermite spline \cite{spitzbart1960generalization}.

%% file: 2-erroranalysis.tex
Now we will examine the error introduced by zeroth order mapping, which is used by existing neuron mapping methods. First, note that under affine transformations, zeroth order mapping of piecewise linear curves introduce no error, so these results are only useful under non-affine transformations. The following results require that the curve $c$ be parameterized by arc length. However, we note that all continuously differentiable regular curves can be reparameterized by arc length \cite{smale1958regular}. We use $\vert \cdot \vert$ to denote the Euclidean norm for elements of $\mathbb{R}^d$, and the spectral norm for matrices.





\begin{proposition}
\textbf{[Zeroth Order Mapping Error Bound]} Say $\phi: \mathbb{R}^3 \rightarrow \mathbb{R}^3$ is a $C^1$ diffeomorphism and $c: [0,L] \rightarrow \mathbb{R}^3$ is a continuous, piecewise linear curve parameterized by arc length with knots $\{t_i: t_1=0, t_n=L, t_{i-1} < t_i\}_{i=1}^n$. For the transformed curve $f=\phi \circ c$, the zeroth order mapping defines a first order spline $g$ which satisfies:

\begin{align}
    \max_{t \in [0,L]} \vert f(t)-g(t)\vert &\leq \max_{i \in \{0,...,n\}, t \in [t_{i-1}, t_i]} \frac{1}{2} \left( \vert D\phi \circ c(t) - I \vert \vert t_i-t_{i-1} \vert + \vert \epsilon_i - \epsilon_{i-1}\vert \right) \label{eq:error2}
\end{align}

\noindent
where $\epsilon_i\triangleq c(t_i)-\phi(c(t_i))$ and $D\phi \circ c(t)$ is the Jacobian of $\phi$ evaluated at $c(t)$.
\end{proposition}


This shows how the error introduced by the state of the art mapping method is related to the displacement magnitude, $\epsilon$, and the extent to which the Jacobian of the transformation, $D\phi$, differs from the identity matrix. Note that the bound in Eq. \ref{eq:error2} goes to zero as $\phi$ approaches the identity map (in which case zeroth order mapping has zero error for piecewise linear curves). It depends on the arc lengths of the original curve segments and the spectral norm of $D\phi$, which is related to the finite time Lyapunov exponent ($\log \vert D\phi\vert$), a well-known quantity in field dynamics which characterizes the amount of stretching in a differentiable flow. Also, the bound applies to $\max_{t \in [0,L]} \vert f(t)-g(t)\vert$, which is not parameterization invariant, and therefore not a strictly geometric quantity. However we note that this quantity is an upper bound of the Frechet distance, which is parameterization invariant. 

In this paper we demonstrate first order mapping in an effort to mitigate this mapping error. Such a method has the advantage of having superior error convergence at the knots as a consequence of Taylor's theorem. Further, we present a set of error bounds that helps clarify the advantage of first order mapping.

\begin{proposition}
\textbf{[Comparable Bounds for Zeroth and First Order Mapping]} Say $\phi: \mathbb{R}^3 \rightarrow \mathbb{R}^3$ is a $C^4$ diffeomorphism and $c: [a,b] \rightarrow \mathbb{R}^3$ is a continuous, piecewise $C^4$ curve parameterized with knots $\{t_i: t_1=a, t_n=b, t_{i-1} < t_i\}_{i=1}^n$. For the transformed curve $f=\phi \circ c$ defined by coordinate functions $f=(f^0,f^1,f^2)^T$, the zeroth order mapping defines a first order spline $g_0$ which satisfies:

\begin{align}
    \max_{t \in [a,b]} \vert f(t)-g_0(t)\vert &\leq \frac{\sqrt{3}}{4} \max_{t \in [a,b], j \in \{0,1,2\}} \vert \partial^{(4)}_t f^j(t)\vert \left(\frac{\delta}{2}\right)^4 + \nonumber \\
    &\frac{\sqrt{3}}{2} \left(\frac{\delta}{2}\right)^2 \max_{i \in \{1...n\}, j \in \{0,1,2\}} \vert \partial^{(3)}_t f^j(t_i)\vert \left(\frac{\delta}{2}\right) + \nonumber \\
    & \frac{\sqrt{3}}{2} \left(\frac{\delta}{2}\right)^2 \max_{i \in \{1...n\}, j \in \{0,1,2\}} \vert \partial^{(2)}_t f^j(t_i)\vert \label{eq:compare0}
\end{align}

\noindent 
where $\delta\triangleq \max_{2 \leq i\leq n} \vert t_i-t_{i-1}\vert$ and $\partial^{(k)}_t f^j(t)$ is the $k$'th derivative of $f^j$ evaluated at $t$. Also, the first order mapping defines a third order spline $g_1$, which satisfies

\begin{align}
    \max_{t \in [a,b]} \vert f(t)-g_1(t)\vert &\leq \frac{\sqrt{3}}{4!} \max_{t \in [a,b], j\in \{0,1,2\}} \vert \partial^{(4)}_t f^j(t)\vert \left(\frac{\delta}{2}\right)^4 \label{eq:compare1}
\end{align}

\noindent

and we note that the bound in \ref{eq:compare1} is tighter than the bound in \ref{eq:compare0}.
Further, there exists a transformed curve $f$ and a set of knots $\{t_i\}_{i=1}^n$ that achieves both bounds exactly.

\end{proposition}

Thus, we have made a connection between the state of the art (zeroth order mapping) and a higher order method (first order mapping) via worst-case bounds on mapping error. The error bound for first order mapping is smaller than that for zeroth order mapping, though for any given curve, either method may produce smaller error than the other. Proofs for the propositions are in the supplement.




%% file: 3-softwaremethod.tex
We implemented a first order discrete mapping method in our our open-source
Python package brainlit. In accordance with original SWC formulation \cite{stockley1993system, cannon1998line}, we compute one-sided derivatives at the knots of the curve from first order splines. Then, once the knot positions and derivatives are transformed, we generate a continuous curve in the new space using Hermite interpolation. Further details of our implementation can be found in the Methods.

Figure \ref{fig:toy} shows examples of our method on simulated data, compared to the zeroth order method, and the ``ground truth'' where we map a dense sampling of points along the first order spline of the original curve.

\begin{figure}[ht]
\centering
\includegraphics[width=\textwidth]{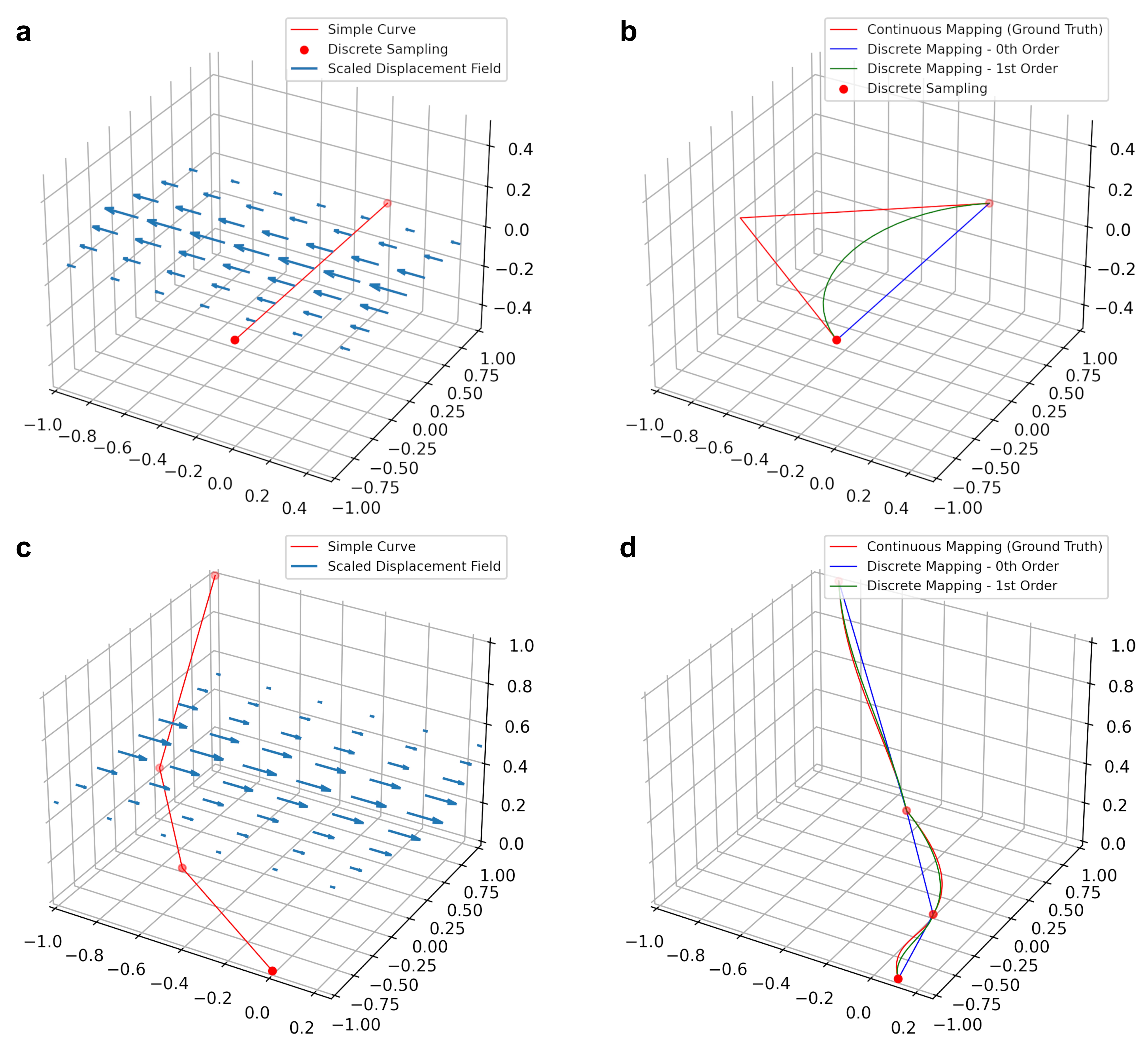}
\caption{Preserving derivative information can mitigate errors when transforming discretized curves. \textbf{a-b} Applying a nonlinear deformation field to a single line segment (\textbf{a}) using zeroth and first order mapping (\textbf{b}). \textbf{c-d} Applying a nonlinear deformation field to a piecewise linear curve (\textbf{c}) using zeroth and first order mapping (\textbf{d}). Zeroth and first order discrete mapping methods are shown relative to ground truth considered to be the application of the vector field to a dense sampling of the original curves.}
\label{fig:toy}
\end{figure}

%% file: 5-realneurons.tex
\def\sigone{80}
\def\sigtwo{160}
\def\sigthree{320}
\def\sigfour{640}

We applied our method to $20$ reconstructed neurons in SWC format from a whole mouse brain image from the Janelia MouseLight project \cite{winnubst2019reconstruction}. We selected the first $20$ SWC files that successfully downloaded from MouseLight's NeuronBrowser repository and did not have repeat trace nodes. Neurons have a tree-like  structure, and we split them into non-branching curves in order to apply our mapping methods. We follow a method introduced previously \cite{athey2021spline} where the root to leaf path with the longest arc length is recursively removed until the tree is reduced to non-bifurcating ``branches''. 

We generate random transformations using the Large Deformation Diffeomorphic Metric Mapping (LDDMM) framework described in Miller et al. and applied in Tward and Miller \cite{miller_geodesic_2006, tward_complexity_2017}. We generate an initial momentum field by sampling Gaussian noise with zero mean and varying standard deviation, $\sigma$. The momentum is smoothed to construct a velocity field, and integrated in time according to the conservation laws established in Miller et al. to generate a diffeomorphic transformation \cite{miller_geodesic_2006}. We generated four diffeomorphisms with $\sigma$ levels of $\sigone$, $\sigtwo$, $\sigthree$ and $\sigfour$ $\mu m/\text{time}$. The position and tangent displacement profiles of these four diffeomorphisms are shown in Figure \ref{fig:results}a. We centered the neuron traces at the origin then applied the random diffeomorphisms to compare zeroth and first order mapping to ground truth (Fig. \ref{fig:results}b-g). Ground truth was generated by upsampling the original traces to a maximum node spacing of $2 \mu m$ followed by zeroth order mapping.

\begin{figure}[ht]
\centering
\includegraphics[width=\textwidth]{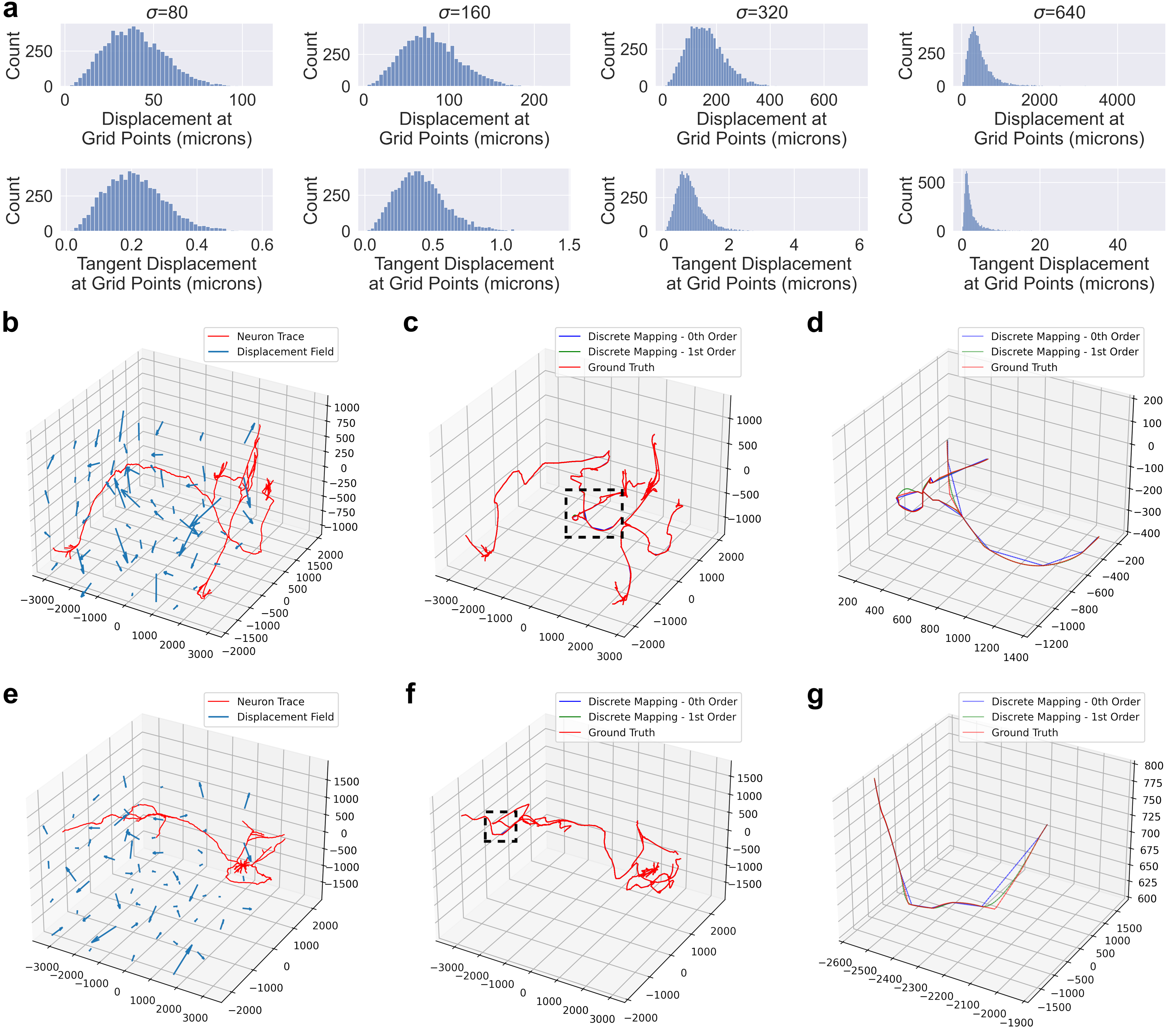}
\caption{Application of zeroth and first order mapping of neuron traces under diffeomorphisms derived from random Gaussian initial momenta. \textbf{a} Different values of $\sigma$ produced diffeomorphisms with different position and tangent vector displacement profiles. The positions and tangent vectors sampled in the histogram were distributed as a uniform grid with a spacing of $500 \mu m$. \textbf{b-g} Two examples of the diffeomorphism with $\sigma=\sigfour$ applied to neuron traces to produce zeroth and first order mappings, along with ground truth. Both examples show the original trace and the transformation \textbf{(b,e)}, the results of the different transformation methods \textbf{(c,f)}, and a zoomed in view of the region outlined by the dotted line to show discrepancies between the methods. Plot axes are in units of microns.}
\label{fig:results}
\end{figure}

For each neuron trace, we computed the discrete frechet error from ground truth (Fig. \ref{fig:stats}a). We also wanted to measure which mapping method better matched the ground truth with respect to a neuron's distribution of common morphometric quantities, such as path angle, branch angle, tortuosity, and segment length. We used the Kolmogorov-Smirnov test statistic to measure how much the distribution of these quantities differed from ground truth (Fig. \ref{fig:stats}b). We performed two-sided Wilcoxon signed-rank tests for each comparison and used a Bonferroni correction across the different $\sigma$ values (Fig. \ref{fig:stats}b). Lastly, we compared the discrete frechet errors the average sampling period of the trace i.e. the average distance between trace nodes (Figure \ref{fig:stats}c).

\begin{figure}[ht]
\centering
\includegraphics[width=\textwidth]{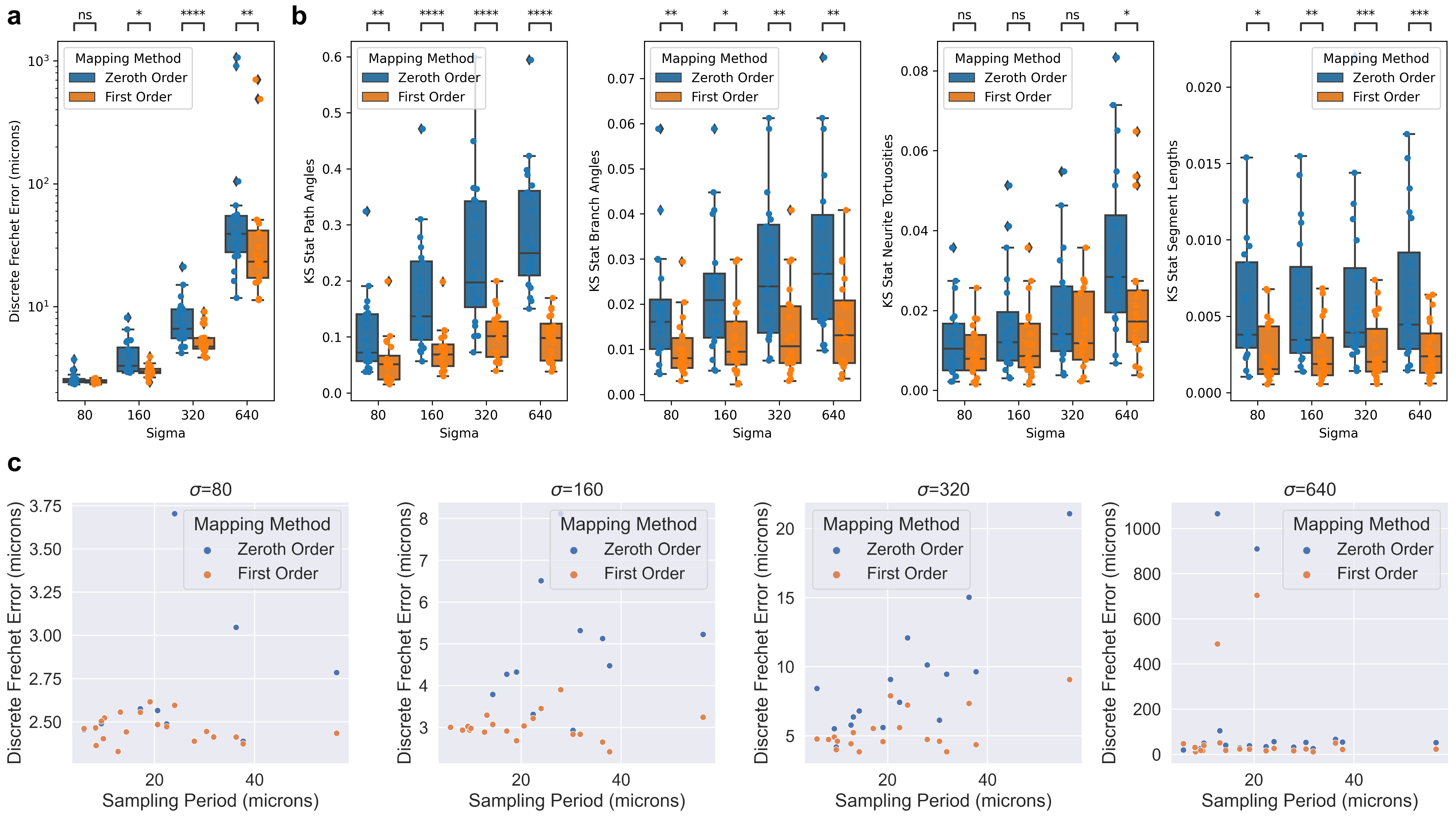}
\caption{Comparison of zeroth and first order mapping of neuron traces under random diffeomorphisms. \textbf{a} Discrete Frechet error was computed between the different order mappings, and ground truth. \textbf{b} Distributions of common morphometric quantities were compared to that of ground truth using the Kolmogorov-Smirnov test statistic. Differences between zeroth and first order methods were tested using Wilcoxon signed-rank test with Bonferroni correction across different values of $\sigma$ ($\ast: p \leq 0.05, \ast \ast: p \leq 0.01, \ast \ast \ast: p \leq 0.001, \ast\ast\ast\ast\: p\leq 0.0001$). Box plots show median, upper and lower quartiles and whiskers have a maximum length of 1.5x the interquartile range with other outlier data marked with points. \textbf{c} Relationship between discrete frechet error and average sampling period (distance between trace points) under the random diffeomorphisms.}
\label{fig:stats}
\end{figure}

To explore the effect of downsampling neuron traces on mapped morphologies, we identified non-branching nodes in straight portions of the trace, and measured the impact of removing those nodes from the trace. Specifically, we performed first order mapping on the segment with the node removed, and compared it to the ground truth mapping of the original segment. We determined which fraction of nodes maintained a discrete frechet error less than one micron, serving as an estimate for the fraction of nodes which are not necessary to maintained the mapped morphology (Fig. \ref{fig:removal}).

\begin{figure}[ht]
\centering
\includegraphics[width=\textwidth]{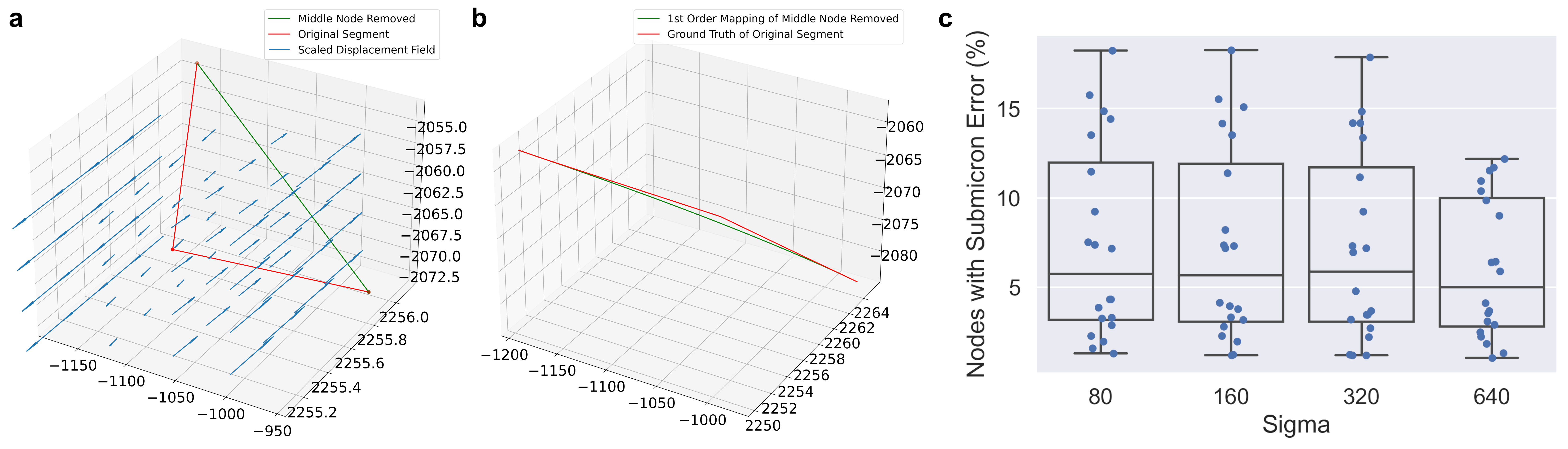}
\caption{Counting how many nodes in MouseLight neuron traces can be removed without affecting the mapped morphology. \textbf{a} For each non-branching node with path angle above 170 degrees, we generated a line segment with that node removed. \textbf{b} We performed first order mapping on the downsampled line segment and compared the result with the ground truth mapping of the original curve. \textbf{c} For each neuron trace, we determined the fraction of nodes where the discrete frechet error is less than or equal to one micron under the four random diffeomorphisms. Box plots show median, upper and lower quartiles and whiskers have a maximum length of 1.5x the interquartile range with other outlier data marked with points.}
\label{fig:removal}
\end{figure}

%% file: 6-discussion.tex
\def\sigone{80}
\def\sigtwo{160}
\def\sigthree{320}
\def\sigfour{630}

In this paper we examine the ``naive'' approach to mapping discretely sampled one-dimensional structures by simply transforming the positions of the knots, i.e. mapping line segments to line segments. We show that this method can be inaccurate when the Jacobian of the transformation is non-constant. We describe how to preserve derivative information which will lead to more accurate mappings in neighborhoods of the knots. We offer an implementation of a first-order mapping technique which, empirically, is more accurate on discretely sampled differentiable curves. We also apply our method to real neuron reconstructions and show that it more accurately matches ground truth in both frechet error, and a variety of morphometric quantities.

In our experiment with real neuron reconstructions, it is important to note what we are considering ground truth. Since the original reconstructions are in SWC format, only the knot positions are known, and the neurons are typically represented as piecewise linear structures. Real neuron morphologies are not piecewise linear, and instead are continuously curving as they pass through dense brain tissue. Nonetheless, because we have no further information about the neuron trajectories, we consider the original reconstructions to be piecewise linear, and generate the ground truth mappings by transforming the straight lines between the knots. 

The transformations in our experiments were generated by ``shooting'' a random initial momenta field \cite{miller_geodesic_2006}. In neuromorphology studies, transformations are typically generated via image registration to an atlas for which several approaches exist \cite{toga_role_2001, chandrashekhar2021cloudreg}. This work is only relevant to non-affine registration techniques since affine transformations preserve straight lines. The utility of higher order mapping depends on the extent to which the brain sample is deformed nonlinearly. In practice, investigators can look at the profiles of position and tangent vector displacements to identify which regime ($\sigma$ level) is most similar to their transformation (Fig. \ref{fig:results}a). At low values of $\sigma$, Frechet error of both zeroth and first order methods are in the range of $1-10$ microns (Fig. \ref{fig:results}c), which is likely negligble for mesoscale neuromorphology. However, under more extreme transformations, the first order mapping offers a more significant improvement in both Frechet error and distributions of morphometric quantities (Fig. \ref{fig:results}c,d).

As mentioned previously, existing mapping methods use zeroth order mapping. Investigators can use the error bound in Eq. \ref{eq:error2} to determine whether zeroth order mapping is adequate. If Jacobian and displacement values of the transformation at hand are not easily accessible, our empirical results can offer guidance. For example, we found that under less extreme transformations ($\sigma=\sigone,\sigtwo$), the frechet errors remained below ten microns for both zeroth and first order methods. However, as transformations got more extreme, it became more important to either keep the sampling period small, or to use first order mapping. Specifically, if the sampling period was less than ten microns, then both zeroth and first order mapping had low error. For higher sampling periods, first order mapping offered more significant improvements.

Conversely our results can be used to make manual tracing more efficient. If the registration transformation, $\phi$, is know a priori, and there are stretches where a neuronal branch is straight, then it is possible to compute the minimum sampling rate while still controlling the amount of error introduced during mapping to atlas coordinates. The neuron trace files examined here are at most a couple megabytes, so this approach is not likely produce significant data storage gains. However, it could allow manual tracers to sample more sparsely along straight stretches of axons, possibly leading to faster reconstruction. As a preliminary experiment, we computed the fraction of nodes which could be removed, while maintaining a submicron error after first order mapping (Fig. \ref{fig:removal}). On average, over 5\% of nodes achieved submicron error, though this fraction decreased with larger sigma, indicating the importance of a higher sampling rate under more extreme transformations. It is important to note that since each node was examined individually, it is not necessarily the case that removing all the nodes together would maintain submicron error. In the worst case, if all the nodes were located consecutively along the trace, only every other node could be removed to maintain submicron error. Further, it is unknown whether skipping the nodes identified in our experiment would have saved time in the MouseLight tracing protocol. A proper experiment to test this hypothesis would involve both registration and neuron reconstruction in real whole-brain images and thus is reserved as a potential avenue of future study. However, given that manual tracing remains a bottleneck and requires several person-hours per neuron \cite{winnubst2019reconstruction}, making tracing process just a couple percentage points faster would tangibly accelerate neuromorpholgical experiments. 

It may be tempting to use our ``ground-truth'' mapping method, i.e. upsampling a linear interpolation then performing zeroth order mapping, as a neuron mapping method. While this may be appropriate in some settings, this approach has two primary disadvantages. First, as stated before, neurons are not piecewise linear structures so, while the knot positions can be generally regarded as lying on the neuron, the linear interpolation cannot. Therefore, it would be necessary to keep track of which knots are from the original trace, and which knots are from the upsampling in order to preserve the original trace information. This would require existing file formats to expand their metadata conventions. Secondly, for large traces, the upsampled data could become computationally cumbersome to store.

Lastly, we want to highlight work in the adjacent field of neuron reconstruction where algorithms such as \cite{li2020brain} can convert reconstruction knots into dense image segmentations which capture neuron trajectories at finer resolutions. Algorithms to automatically trace images of single neurons have been under development for decades \cite{peng2015diadem, athey2022hidden}. They could be adapted to generate both denser neuron samplings, and more accurate derivative estimates at the sampled points. These methods could improve both zeroth and first order mapping methods, so weighing these effects alongside the accuracy required for the given scientific goal would help determine which mapping method is appropriate.

%% file: methods.tex
\subsection*{Software Implementation}

In order to implement a first order action method that transforms neuronal curves, we needed to address two questions. The first is how to estimate derivatives in the original discretely sampled curve. The second is, once the knot positions and derivatives are transformed, how can they be used to generate a continuous curve in the new space.

\subsubsection*{One Sided Derivatives from First Order Splines}

The original trace points are assumed to represent the knots in a first order spline, i.e. the points are linearly interpolated. In this representation, derivatives do not necessarily exist at the knots, but one-sided derivatives do and can be easily computed using the difference of consecutive knot positions. Both one-sided derivatives ($\dot c(t_i^-),\dot c(t_i^+)$) are transformed according to Eq. \ref{eq:action} and used to generate the transformed curve.






\subsubsection*{Fitting Curve to Transformed Positions and Derivatives}

For the $i$'th curve segment, the positions $c(t_{i-1}), c(t_i)$ and one sided derivatives $c(t_{i-1}^+), c(t_i^-)$
present four constraints for the interpolating curve. We use these constraints to define a cubic polynomial between each pair of knots, which is known as Hermite interpolation \cite{kincaid2009numerical}. The result is a third order spline. Specifically, we use the scipy implementation of cubic Hermite splines \cite{2020SciPy-NMeth}. It is important to note that this spline is still not necessarily differentiable at the knots.

\subsection*{Quantitatively Comparing Curves}

As described in the Results, the ground truth was considered to be the zeroth order mapping of sampling every $2$ microns along the piecewise linear trace. The splines defined by zeroth and first order mapping were sampled at the same values of the independent variable as the ground truth (every $2$ microns of arc length before the transformation) then compared with ground truth. We used the package from Jekel et al. to compute discrete frechet distance \cite{Jekel2019}. Discrete Frechet distance is an approximation of, and upper bound to Frechet distance \cite{eiter1994computing}. We used nGauge to compute morphometric quantities and SciPy to perform Kolmogorov-Smirnov statistics \cite{walker_ngauge_2022, 2020SciPy-NMeth}.

Further details about our implementation can be found in our open-source Python package brainlit: \\
http://brainlit.neurodata.io/.

%% file: supplement.tex
\section*{Mathematical Proofs}

We use $\vert \cdot \vert$ to denote the Euclidean norm for elements of $\mathbb{R}^d$, and the spectral norm for matrices.

\begin{statement}
For a sequence of time-stamped elements on the jet space, 
$T = \{(t_i, x_i^{(k)})\}_{i=1}^n$ in $(J^k)^n$, we define the action of diffeomorphisms

\begin{align}
    \phi \cdot T = \{(t_i, \phi\cdot x_i^{(k)})\}_{i=1}^n \label{sup-eq:action}
\end{align}
\end{statement}
\begin{proof}
For a regular differentiable curve $c:[0,L]\rightarrow \mathbb{R}^3$ with extension $\hat{c}:[0,L]\rightarrow X^{(k)}$, we defined $\phi \cdot \hat{c}$ as the extension, $\widehat{\phi \circ c}$. In practice, only the finite sampling $\{x^{(k)}_i\}_{i=1}^n$ is accessible. However, it is always possible to define a curve $c$ that agrees with this sampling, such as with a polynomial \cite{olver1995equivalence}. Then, we can apply $\phi$ to this curve, and compute the transformed positions and derivatives, defining the action on the finite sampling. Now, we verify the axioms of a group action. 



First, the identity element ($\phi_{Id}$) in the diffeomorphism group should leave a sampling unchanged. Assume that $c$ is a curve that agrees with the sampling $T$:

\begin{align*}
    \phi_{Id} \cdot T &= \{(t_i, \bar x_i^{(k)}) : \bar x_i^{(k)} = ((\phi_{Id} \circ c)(t_i),\partial_t (\phi_{Id} \circ c)(t_i),...,\partial_t^k (\phi_{Id} \circ c)(t_i))\}_{i=1}^n \\
    &= \{(t_i, \bar x_i^{(k)}) : \bar x_i^{(k)} = (c(t_i),\partial_t c(t_i),...,\partial_t^k c(t_i))\}_{i=1}^n \\
    &= T
\end{align*}

Second, a composition of diffeomorphisms ($\phi_1 \circ \phi_2$) should act successively on a sampling:

\begin{align*}
    (\phi_1 \circ \phi_2) \cdot T &=  \{(t_i, \bar x_i^{(k)}) : \bar x_i^{(k)} = ... \\
    &((\phi_1 \circ \phi_2 \circ c)(t_i),\partial_t (\phi_1 \circ \phi_2 \circ c)(t_i),...,\partial_t^k (\phi_1 \circ \phi_2 \circ c)(t_i))\}_{i=1}^n \\
    &=  \{(t_i, \bar x_i^{(k)}) : \bar x_i^{(k)} = ... \\
    &((\phi_1 \circ f)(t_i),\partial_t (\phi_1 \circ f)(t_i),...,\partial_t^k (\phi_1 \circ f)(t_i))\}_{i=1}^n & f \triangleq \phi_2 \circ c \\
    &= \phi_1 \cdot \{(t_i, y_i^{(k)}) : y_i^{(k)} = (f(t_i),\partial_t f(t_i),...,\partial_t^k f(t_i)) \}_{i=1}^n \\
    &= \phi_1 \cdot \{(t_i, y_i^{(k)}) : y_i^{(k)} = ... \\
    &((\phi_2 \circ c)(t_i),\partial_t (\phi_2 \circ c)(t_i),...,\partial_t^k (\phi_2 \circ c)(t_i))\}_{i=1}^n \\
    &= \phi_1 \cdot \phi_2 \cdot \{(t_i, x_i^{(k)}) : x_i^{(k)} = (c(t_i),\partial_t c(t_i),...,\partial_t^k c(t_i)) \}_{i=1}^n \\
    &= \phi_1 \cdot \phi_2 \cdot \{(t_i, x_i^{(k)}) : x_i^{(k)} = (c(t_i),\partial_t c(t_i),...,\partial_t^k c(t_i)) \}_{i=1}^n \\
    &= \phi_1 \cdot \phi_2 \cdot T
\end{align*}
\end{proof}

\begin{proposition}
\textbf{[Zeroth Order Mapping Error Bound]} Say $\phi: \mathbb{R}^3 \rightarrow \mathbb{R}^3$ is a $C^1$ diffeomorphism and $c: [0,L] \rightarrow \mathbb{R}^3$ is a continuous, piecewise linear curve parameterized by arc length with knots $\{t_i: t_1=0, t_n=L, t_{i-1} < t_i\}_{i=1}^n$. For the transformed curve $f=\phi \circ c$, the zeroth order mapping defines a first order spline $g$ which satisfies:

\begin{align}
    \max_{t \in [0,L]} \vert f(t)-g(t)\vert &\leq \max_{i \in \{0,...,n\}, t \in [t_{i-1}, t_i]} \frac{1}{2} \left( \vert D\phi \circ c(t) - I \vert \vert t_i-t_{i-1} \vert + \vert \epsilon_i - \epsilon_{i-1}\vert \right) \label{sup-eq:error2}
\end{align}

\noindent
where $\epsilon_i\triangleq c(t_i)-\phi(c(t_i))$ and $D\phi \circ c(t)$ is the Jacobian of $\phi$ evaluated at $c(t)$.
\end{proposition}

\begin{proof}
We will focus on a single line segment $c_i=c\vert_{[t_{i-1},t_i]}$, then maximize over all such segments. $c_i$ is a function from $[t_{i-1},t_i]$ to $\mathbb{R}^3$. Denote the endpoints of $c_i$ as $c_{i,0}=c_i(t_{i-1})$ and $c_{i,1}=c_i(t_i)$. The zeroth order mapping of $c_i$ defines the first order spline $g_{c_i}(t)=\phi(c_{i,0}) + \frac{(t-t_{i-1})}{t_i-t_{i-1}} (\phi(c_{i,1})-\phi(c_{i,0}))$.

For simplicity, we will reparameterize the problem using $\sigma(t)=t_{i-1}+t(t_i-t_{i-1}):[0,1]\rightarrow [t_{i-1},t_i]$ and define $c'=c_i \circ \sigma$ which is defined on $[0,1]$. The zeroth order mapping of $c'$ defines the spline $g_{c'}(t)=\phi(c_{i,0}) + t (\phi(c_{i,1})-\phi(c_{i,0}))$. Note that the zeroth order mapping errors are the same in both parameterizations i.e. for $f_{c_i}=\phi \circ c_i, f_{c'}=\phi \circ c'$ we have:

\begin{align*}
    \max_{t \in [t_{i-1},t_i]} \vert f_{c_i}(t)-g_{c_i}(t)\vert &=\max_{t \in [0,1]} \vert f_{c'}(t)-g_{c'}(t)\vert
\end{align*}

since, for every $t \in [0,1]$, $\vert f_{c'}(t)-g_{c'}(t)\vert = \vert f_{c_i}(\sigma(t))-g_{c_i}(\sigma(t))\vert$ and for every $t \in [t_{i-1},t_i]$, $\vert f_{c_i}(t)-g_{c_i}(t)\vert = \vert f_{c'}(\sigma^{-1}(t))-g_{c'}(\sigma^{-1}(t))\vert$. So, we have converted the problem to bounding:

\begin{align*}
    \max_{t \in [0,1]} \vert f_{c'}(t)-g_{c'}(t)\vert
\end{align*}

We have

\begin{align*}
    f_{c'}(t)-g_{c'}(t) &= \phi(c'(t)) - \left[\phi(c_{i,0}) + t (\phi(c_{i,1})-\phi(c_{i,0})) \right]
\end{align*}

and since $f_{c'}(t)-g_{c'}(t)$ vanishes at both $t=0$ and $t=1$, the following argument, which uses the fundamental theorem of calculus, applies both going forward from $t=0$ and backward from $t=1$. So, without loss of generality, we consider $0 \leq t \leq \frac{1}{2}$:

\begin{align*}
    f_{c'}(t)-g_{c'}(t) &= \phi(c'(t)) - \left[\phi(c_{i,0}) + t (\phi(c_{i,1})-\phi(c_{i,0})) \right]\\
    &= \int_0^t \partial_{\tau} \left(\phi(c'(\tau)) - \left[\phi(c_{i,0}) + \tau (\phi(c_{i,1})-\phi(c_{i,0})) \right] \right) d\tau \\
    &= \int_0^t D\phi \circ c'(\tau)\cdot \dot c'(\tau) - (\phi(c_{i,1})-\phi(c_{i,0})) d \tau \\
    &\leq \max_{t \in [0,1]} \vert D\phi \circ c'(t) \cdot \dot c'(t) - (\phi(c_{i,1})-\phi(c_{i,0})) \vert \int_0^t d\tau \\
    &\leq \frac{1}{2} \max_{t \in [0,1]} \vert D\phi \circ c'(t) \cdot \dot c'(t) - (\phi(c_{i,1})-\phi(c_{i,0})) \vert &t\leq \frac{1}{2} \\
    &\text{Define: }\epsilon_i\triangleq c_{i,1} - \phi(c_{i,1}),\epsilon_{i-1}\triangleq c_{i,0} - \phi(c_{i,0}) \\
    &= \frac{1}{2} \max_{t \in [0,1]} \vert D\phi \circ c'(t) \cdot (c_{i,1}-c_{i,0}) - (c_{i,1}-c_{i,0}) + (\epsilon_i - \epsilon_{i-1}) \vert  \\
    &\leq \max_{t \in [0,1]} \frac{1}{2}\left( \vert D\phi \circ c'(t) - I \vert \vert c_{i,1}-c_{i,0} \vert  +\vert \epsilon_i - \epsilon_{i-1} \vert \right) \\
    &= \max_{t \in [t_{i-1},t_i]} \frac{1}{2}\left(\vert D\phi \circ c(t) - I \vert \vert c_{i,1}-c_{i,0} \vert  +\vert \epsilon_i - \epsilon_{i-1} \vert \right)\\
    &= \max_{t \in [t_{i-1},t_i]} \frac{1}{2}\left(\vert D\phi \circ c(t) - I \vert \vert t_i-t_{i-1} \vert  +\vert \epsilon_i - \epsilon_{i-1} \vert\right)
\end{align*}

where the last equality comes from the fact that $c$ is parametrized by arc length, so $\vert t_i-t_{i-1} \vert=\vert c_{i,1}-c_{i,0} \vert$. In summary we have:

\begin{align*}
    \max_{t \in [t_{i-1},t_i]} \vert f_{c_i}(t)-g_{c_i}(t)\vert &\leq \max_{t \in [t_{i-1},t_i]} \frac{1}{2}\left(\vert D\phi \circ c(t) - I \vert \vert t_i-t_{i-1} \vert  +\vert \epsilon_i - \epsilon_{i-1} \vert\right)
\end{align*}

Finally, we maximize over all segments to get:

\begin{align*}
    \max_{t \in [0,L]} \vert f(t)-g(t)\vert &\leq \max_{i \in \{0,...,n\}, t \in [t_{i-1},t_i]} \frac{1}{2}\left(\vert D\phi \circ c(t) - I \vert \vert t_i-t_{i-1} \vert  +\vert \epsilon_i - \epsilon_{i-1} \vert\right)
\end{align*}

where $f=\phi \circ c$ and $g$ is the first order spline defined by the zeroth order mapping of $c$. 

\end{proof}

\begin{proposition}
\textbf{[Comparable Bounds for Zeroth and First Order Mapping]} Say $\phi: \mathbb{R}^3 \rightarrow \mathbb{R}^3$ is a $C^4$ diffeomorphism and $c: [a,b] \rightarrow \mathbb{R}^3$ is a continuous, piecewise $C^4$ curve parameterized with knots $\{t_i: t_1=a, t_n=b, t_{i-1} < t_i\}_{i=1}^n$. For the transformed curve $f=\phi \circ c$ defined by coordinate functions $f=(f^0,f^1,f^2)^T$, the zeroth order mapping defines a first order spline $g_0$ which satisfies:

\begin{align}
    \max_{t \in [a,b]} \vert f(t)-g_0(t)\vert &\leq \frac{\sqrt{3}}{4} \max_{t \in [a,b], j \in \{0,1,2\}} \vert \partial^{(4)}_t f^j(t)\vert \left(\frac{\delta}{2}\right)^4 + \nonumber \\
    &\frac{\sqrt{3}}{2} \left(\frac{\delta}{2}\right)^2 \max_{i \in \{1...n\}, j \in \{0,1,2\}} \vert \partial^{(3)}_t f^j(t_i)\vert \left(\frac{\delta}{2}\right) + \nonumber \\
    & \frac{\sqrt{3}}{2} \left(\frac{\delta}{2}\right)^2 \max_{i \in \{1...n\}, j \in \{0,1,2\}} \vert \partial^{(2)}_t f^j(t_i)\vert \label{sup-eq:compare0}
\end{align}

\noindent 
where $\delta\triangleq \max_{2 \leq i\leq n} \vert t_i-t_{i-1}\vert$ and $\partial^{(k)}_t f^j(t)$ is the $k$'th derivative of $f^j$ evaluated at $t$. Also, the first order mapping defines a third order spline $g_1$, which satisfies

\begin{align}
    \max_{t \in [a,b]} \vert f(t)-g_1(t)\vert &\leq \frac{\sqrt{3}}{4!} \max_{t \in [a,b], j\in \{0,1,2\}} \vert \partial^{(4)}_t f^j(t)\vert \left(\frac{\delta}{2}\right)^4 \label{sup-eq:compare1}
\end{align}

\noindent

and we note that the bound in \ref{eq:compare1} is tighter than the bound in \ref{eq:compare0}.
Further, there exists a transformed curve $f$ and a set of knots $\{t_i\}_{i=1}^n$ that achieves both bounds exactly.

\end{proposition}
\begin{proof}
    We will prove both bounds starting with a single segment of $c$, then extending to the entire piecewise curve.
    
    The bound in \ref{sup-eq:compare1} comes from the error estimate of Hermite interpolation for one dimensional functions, Theorem 2 in Section 6.3 of \cite{kincaid2009numerical}. For a single segment between knots $t_{i-1},t_i$, this theorem states that if $p$ is the polynomial of degree at most $3$ which agrees with $h$ and $\partial_t h$ at the knots, then, for each $t$, there exists a point $\xi \in (t_{i-1},t_i)$ such that:

    \begin{align*}
        h(t)-p(t) &= \frac{\partial_t^{(4)} h(\xi)}{4!}(t-t_{i-1})^2(t-t_i)^2 \\
        &\text{therefore, for all }t \\
        \vert h(t)-p(t) \vert &\leq \frac{1}{4!} \max_{t \in [t_{i-1},t_i]} \vert \partial_t^{(4)}h(t) \vert \left( \frac{t_i-t_{i-1}}{2}\right)^4
    \end{align*}

    The first order mapping is indeed the third order spline that matches function and derivative values at the knots, so the above bound applies to all three coordinate functions of $f=(f^0,f^1,f^2)^T$ and $g_1=(g_1^0, g_1^1, g_1^2)^T$. To accommodate all three dimensions, we maximize over all dimensions and add a $\sqrt{3}$ term:

    \begin{align*}
        \max_{t \in [t_{i-1},t_i]} \vert f(t)-g_1(t) \vert &\leq \frac{\sqrt{3}}{4!} \max_{t \in [t_{i-1},t_i], j\in \{0,1,2\}} \vert \partial_t^{(4)}f^j(t) \vert \left( \frac{t_i-t_{i-1}}{2}\right)^4
    \end{align*}

    Then if we maximize both sides over all the segments, we get

    \begin{align*}
        \max_{t \in [a,b]} \vert f(t)-g_1(t) \vert &\leq \frac{\sqrt{3}}{4!} \max_{t \in [a,b], j\in \{0,1,2\}} \vert \partial_t^{(4)}f^j(t) \vert \left( \frac{\delta}{2}\right)^4
    \end{align*}

    The bound in \ref{sup-eq:compare0} comes from the error estimate of polynomial interpolation, Theorem 2 in Section 6.1 of \cite{kincaid2009numerical}. This theorem states that, for a single segment between knots $t_{i-1},t_i$, if $p$ is the line that agrees with a one dimensional function $h$ at the knots, then there exists a point $\xi \in (t_{i-1},t_i)$ such that:

    \begin{align}
        h(t)-p(t)&=\frac{\partial_t^{(2)}h(\xi)}{2} (t-t_{i-1})(t-t_i) \nonumber \\
        &\text{therefore} \nonumber \\
         \max_{t \in [t_{i-1,t_i}]}\vert h(t)-p(t) \vert &\leq \frac{1}{2} \max_{t \in [t_{i-1}, t_i]} \vert \partial_t^{(2)}h(t) \vert \left(\frac{t_i-t_{i-1}}{2}\right)^2 \label{sup-eq:polybound}
    \end{align}

    Our remaining task is to relate the maximum second derivative to the maximum fourth derivative. We start by using the fundamental theorem of calculus twice to get (for $t\in [t_{i-1},t_i]$):

    \begin{align}
        \partial_t^{(2)}h(t)&= \int_{t_{i-1}}^t \int_{t_{i-1}}^\tau h^{(4)}(\upsilon) d\upsilon d\tau + \partial_t^{(3)}h(t_{i-1})(t-t_{i-1}) + \partial_t^{(2)} h(t_{i-1}) \nonumber \\
        &\text{therefore} \nonumber \\
        \vert \partial_t^{(2)}h(t) \vert &\leq \max_{t \in [t_{i-1},t_i]} \vert \partial_t^{(4)}h(t) \vert \int_{t_{i-1}}^t \int_{t_{i-1}}^\tau  d\upsilon d\tau + \nonumber \\
        &\vert t-t_{i-1}\vert \max_{t \in \{t_{i-1}, t_i\}} \vert\partial_t^{(3)}h(t) \vert + \max_{t \in \{t_{i-1}, t_i\}} \vert \partial_t^{(2)} h(t) \vert \nonumber \\
        &= \frac{1}{2}(t-t_{i-1})^2 \max_{t \in [t_{i-1},t_i]} \vert \partial_t^{(4)}h(t) \vert  + \nonumber \\
        &\vert t-t_{i-1}\vert \max_{t \in \{t_{i-1}, t_i\}} \vert\partial_t^{(3)}h(t) \vert + \max_{t \in \{t_{i-1}, t_i\}} \vert \partial_t^{(2)} h(t) \vert \label{sup-eq:bound1}
    \end{align}
    Similarly,
    \begin{align}
        \partial_t^{(2)}h(t)&= \int_{t_{i}}^t \int_{t_{i}}^\tau h^{(4)}(\upsilon) d\upsilon d\tau + \partial_t^{(3)}h(t_{i})(t-t_{i}) + \partial_t^{(2)} h(t_{i}) \nonumber \\
        &\text{therefore} \nonumber \\
        \vert \partial_t^{(2)}h(t) \vert &\leq \max_{t \in [t_{i-1},t_i]} \vert \partial_t^{(4)}h(t) \vert \int_{t_{i}}^t \int_{t_{i}}^\tau  d\upsilon d\tau + \nonumber \\
        &\vert t-t_{i}\vert \max_{t \in \{t_{i-1}, t_i\}} \vert\partial_t^{(3)}h(t) \vert + \max_{t \in \{t_{i-1}, t_i\}} \vert \partial_t^{(2)} h(t) \vert \nonumber \\
        &= \frac{1}{2}(t-t_{i})^2 \max_{t \in [t_{i-1},t_i]} \vert \partial_t^{(4)}h(t) \vert  + \nonumber \\
        &\vert t-t_{i}\vert \max_{t \in \{t_{i-1}, t_i\}} \vert\partial_t^{(3)}h(t) \vert + \max_{t \in \{t_{i-1}, t_i\}} \vert \partial_t^{(2)} h(t) \vert \label{sup-eq:bound2}
    \end{align}

    So, we have two bounds for $\vert \partial_t^{(2)}h(t) \vert$, given by \ref{sup-eq:bound1} and \ref{sup-eq:bound2}. Both bounds are in the form of a second order polynomial of $t$, so it is straightforward to show that the combined bound is maximal at the midpoint between $t_{i-1}$ and $t_i$, where the bounds also happen to intersect. Thus, we have:

    \begin{align}
        \max_{t \in [t_{i-1},t_i]} \vert \partial_t^{(2)}h(t) \vert &\leq \frac{1}{2}\left(\frac{t_i-t_{i-1}}{2}\right)^2 \max_{t \in [t_{i-1},t_i]} \vert \partial_t^{(4)}h(t) \vert  + \nonumber \\
        &\left( \frac{t_i-t_{i-1}}{2}\right) \max_{t \in \{t_{i-1}, t_i\}} \vert\partial_t^{(3)}h(t) \vert + \max_{t \in \{t_{i-1}, t_i\}} \vert \partial_t^{(2)} h(t) \vert \label{sup-eq:2dbound}
    \end{align}
    
    Combining \ref{sup-eq:2dbound} with \ref{sup-eq:polybound}, we get

    \begin{align*}
        \max_{t \in [t_{i-1,t_i}]}\vert h(t)-p(t) \vert &\leq\frac{1}{4} \max_{t \in [t_{i-1}, t_i]} \vert \partial_t^{(4)}h(t) \vert \left(\frac{t_i-t_{i-1}}{2}\right)^4 + \\
        & \frac{1}{2} \left( \frac{t_i-t_{i-1}}{2}\right)^3 \max_{t \in \{t_{i-1}, t_i\}} \vert\partial_t^{(3)}h(t) \vert + \\
        &\frac{1}{2} \left( \frac{t_i-t_{i-1}}{2}\right)^2 \max_{t \in \{t_{i-1}, t_i\}} \vert \partial_t^{(2)} h(t) \vert 
    \end{align*}

    Again, to extend this result from a function that takes values in $\mathbb{R}$, to one that takes values in $\mathbb{R}^3$, we need to maximize over all dimensions and include a factor of $\sqrt{3}$. After maximizing over all segments, we get, for a piecewise $C^4$ curve $f=\phi \circ c:[a,b] \rightarrow \mathbb{R}^3$, and the zeroth order mapping $g_0$ which linearly interpolates the mapped knots:

    \begin{align*}
        \max_{t \in [a,b]}\vert f(t)-g_0(t) \vert &\leq \frac{\sqrt{3}}{4} \max_{t \in [a, b], j\in \{0,1,2\}} \vert \partial_t^{(4)}f^j(t) \vert \left(\frac{\delta}{2}\right)^4 + \\
        & \frac{\sqrt{3}}{2} \left( \frac{\delta}{2}\right)^2  \max_{i \in \{0...n\}, j\in \{0,1,2\}} \vert\partial_t^{(3)}f^j(t_i) \vert \left( \frac{\delta}{2}\right) + \\
        & \frac{\sqrt{3}}{2} \left( \frac{\delta}{2}\right)^2 \max_{i \in \{0...n\}, j\in \{0,1,2\}} \vert \partial_t^{(2)} f^j(t_i) \vert
    \end{align*}
    Finally, we note that if the mapped curve is $f(t)=(1-t^2, 1-t^2, 1-t^2)^T$ defined on $[-1,1]$ with two knots $t_0=-1,t_1=1$, then the zeroth order mapping is $g_0(t)=(0,0,0)$ and the first order mapping is $g_1(t)=(1-t^2, 1-t^2, 1-t^2)$. The first order mapping bound satisfies the bound from \ref{sup-eq:compare0} of 0, because $f=g_1$. The zeroth order bound is $\sqrt{3}$ which is achieved by the zeroth order mapping at $t=0$.
\end{proof}









